\documentclass[aps,reprint]{revtex4-1}
\usepackage[utf8]{inputenc}
\usepackage{amsfonts}
\usepackage{amsmath}
\usepackage{amssymb}
\usepackage{amsthm}
\usepackage{mathtools}
\usepackage{enumerate}
\usepackage[shortlabels]{enumitem}
\usepackage{array}
\usepackage{bm}
\usepackage{graphicx}
\usepackage{nicefrac}
\usepackage{booktabs}
\usepackage[all]{xy}
\usepackage{hyperref}
\usepackage{subcaption}
\usepackage{bbm}
\usepackage{dsfont}
\usepackage{tikz}
\usepackage{tikz-cd}
\usepackage{tkz-graph}
\usetikzlibrary{arrows}
\usepackage{cancel}
\usepackage{xcolor}

\usepackage[font=small,labelfont=bf]{caption}
\usepackage[colorinlistoftodos,prependcaption]{todonotes}
\usepackage{MnSymbol}
\usepackage{multirow}

\definecolor{light-gray}{gray}{0.6}

\newcommand\mc[1]{\mathcal{#1}}

\newcommand\PPi{\underline{\Pi}}

\newtheorem{theorem}{Theorem}
\newtheorem*{theorem*}{Theorem}
\newtheorem{lemma}{Lemma}
\newtheorem{proposition}{Proposition}

\newtheorem{definition}{Definition}

\newcolumntype{C}{>{$}c<{$}}
\AtBeginDocument{
\heavyrulewidth=.08em
\lightrulewidth=.05em
\cmidrulewidth=.03em
\belowrulesep=.65ex
\belowbottomsep=0pt
\aboverulesep=.4ex
\abovetopsep=0pt
\cmidrulesep=\doublerulesep
\cmidrulekern=.5em
\defaultaddspace=.5em
}

\begin{document}
\title{No-signalling, Contextuality, and the Arrow of Time}
\author{Markus Frembs$^{1,}$}
\email{markus.frembs13@imperial.ac.uk}
\vspace{-0.5cm}
\author{Andreas D\"oring}
\email{info@andreas-doering.net}
\affiliation{$^1$Blackett Laboratory, Imperial College London, Prince Consort Road, SW7 2AZ, United Kingdom}

\begin{abstract}
    Bell's seminal paper shows that some correlations in quantum theory are not reconcilable with hidden variables and the classical notion of locality. Yet, a weaker notion of locality, known as no-signalling, survives the no-go-result. We study its restrictiveness by considering the full set of local quantum observables. This leads to a much larger set of no-signalling constraints than usually considered, which has been shown to be sufficient to exclude PR-boxes and other types of super-quantum correlations. Here, we emphasise that this result rests on contextuality in a fundamental way, in particular, we show how no-signalling arises naturally from context composition. Based on this connection, we improve existing results by establishing a one-to-one correspondence between quantum states and (collections of) non-signalling probability distributions over product contexts under an additional consistency condition between time arrows in subsystems.
\end{abstract}

\maketitle


\paragraph*{Introduction.} A common version of Bell's theorem is the CHSH inequality, which bounds correlations in the outcome statistics of local Stern-Gerlach experiments on two space-like separated spin-$\frac{1}{2}$ systems \cite{CHSH}. Measuring the spin component along an arbitrary direction in space, one always obtains either one of two outcomes: spin-up or spin-down. A natural notion of locality has that the individual choices of (spin) measurements as well as their outcomes neither affect the choice nor outcome of the (spin) measurement at the space-like separated site. Locality in this form is a feature of classical physics (more generally, deterministic hidden variable models) \cite{Fine1982}, and constrains the expectation value of the quantity $c = a\times b + a\times b' + a'\times b - a'\times b'$, where $a,a'$ and $b,b'$ denote local Stern-Gerlach measurements, within the latter to
\begin{equation}\label{eq: classical CHSH bound}
    \mathbb{E}_\mathrm{cl}(c) \leq 2\; .
\end{equation}
Quantum correlations famously exceed this bound and are themselves subject to Tsirelson's bound \cite{Tsirelson1980}: $\mathbb{E}_\mathrm{qm}(c) \leq 2\sqrt{2}$. A violation of Eq.~(\ref{eq: classical CHSH bound}) in the CHSH experiment has been reported in recent years after previous attempts had been shown to allow for certain loopholes \cite{ShalmEtAl2015,ZeilingerEtAl2015}.
It is now widely accepted that nature is not classical and quantum correlations do exceed those in Eq. (\ref{eq: classical CHSH bound}) as well as higher dimensional, multipartite generalisations thereof. The usual interpretation is that nature is nonlocal.

Given the violation of classical correlation bounds, understanding the limitations on quantum correlations marks an important and ongoing research objective. One line of this research has focused on the fact that while quantum mechanics is hardly reconcilable with the classical notion of locality, it does satisfy the more general notion of no-signalling.
Let $a, b$ represent measurements in systems $1, 2$, respectively, and let $A$, $B$ represent corresponding measurement outcomes. Then the measurement statistics satisfy \emph{no-signalling} if the joint probability distribution marginalises to local probability distributions conditioned on the choice of \emph{local} measurements only, formally:
\begin{align}\label{eq: no-signalling probability distributions}
    \begin{split}
    P(A \mid a) &= \sum_B P(A,B \mid a,b) \\
    P(B \mid b) &= \sum_A P(A,B \mid a,b)
    \end{split}
\end{align}
A well-known example for such a set of (non-physical) probability distributions is given by a PR-box \cite{PopescuRohrlich1994}. It thus seems that no-signalling as a physical principle is too weak to single out quantum theory and other, often information theoretically motivated principles have been suggested to complement it (see \cite{PawlowskiEtAl2009, vanDam2013}, for instance).

Importantly, Eq.~(\ref{eq: no-signalling probability distributions}) ranges over multiple measurement settings. In order to evaluate all constraints, one must therefore specify the possible choices of, e.g. spin measurements on either subsystem first. From this perspective PR-boxes are probability distributions constrained by just two possible measurement settings on either side. A physically more interesting scenario is that in which we allow arbitrary local (spin) measurements. There are then many more constraints inherent to Eq.~(\ref{eq: no-signalling probability distributions}), and one might suspect that no-signalling restricts to quantum states over the set of \emph{all local} quantum observables. This idea has been studied in \cite{Wehner2010} and formerly in \cite{KlayRandallFoulis1987,Wallach2002}, where it was shown that this hypothesis is false.

In this letter we substantially improve the situation by complementing the narrative with a notion of time orientation in (local) subsystems, which will eventually allow us to identify non-signalling probability distributions with quantum states unambiguously.
Our result fundamentally builds on the concept of contextuality, which may come as little surprise since the violation of Eq.~(\ref{eq: classical CHSH bound}) in quantum theory is closely related to the non-existence of a classical state space picture, as famously proven by Kochen and Specker \cite{KochenSpecker1967}. Generalising to contextual systems,
we find that no-signalling arises naturally from context composition. This is already evident in the framework of unentangled frame functions in \cite{Wallach2002}; for clarity, we thus start with a brief review of the latter.
For more details on the relation between contextuality and nonlocality, as well as other key quantum features, we refer to \cite{FreDoe19a}.
Proofs and further details can be found in the Supplementary Material to this letter.




\paragraph*{Non-signalling (product) frame functions.} Throughout, we take (local) observables to be represented by bounded self-adjoint operators on some Hilbert space, $\mc{L}_\mathrm{sa}(\mc{H})$. For simplicity we only consider finite dimensional systems, in particular, $\mc{H} = \mathbb{C}^d$,
$\mc{L}(\mc{H}) = M_d(\mathbb{C})$, and
$\mc{L}_\mathrm{sa}(\mc{H}) = \{a \in M_d(\mathbb{C}) \mid a^* = a\}$.

Recall that a \emph{frame function of weight $W \in \mathbb{R}$} on the unit sphere $S(\mc{H})$ is a function $f: S(\mc{H}) \rightarrow \mathbb{R}$ such that $\sum_{j=1}^d f(v_j) = W$ for all orthonormal bases $(v_j)_{j=1}^d \in \mathrm{ONB}(\mc{H})$. Frame functions play a crucial role in Gleason's theorem \cite{Gleason1975}, and in its generalisation to composite systems in \cite{KlayRandallFoulis1987, Wallach2002}.

Note that from an operational perspective the only outcomes accessible to local observers correspond to elements in $\sigma(\mc{H}) := \{v_1 \otimes \cdots \otimes v_n \in S(\mc{H}) \mid v_i \in S(\mc{H}_i)\}$, $\mc{H} = \otimes_{i=1}^n \mc{H}_i$. It is thus natural to consider non-negative \emph{unentangled frame functions} with domain $\sigma(\mc{H}) \subsetneq S(\mc{H})$ and constraints restricted to $\mathrm{ONB}(\sigma(\mc{H})) \subsetneq \mathrm{ONB}(\mc{H})$ instead.
\cite{Wallach2002} shows that such frame functions almost correspond with quantum states in the form of density matrices, where almost means up to positivity (cf. \cite{KlayRandallFoulis1987, Wehner2010}).

A further restriction compared to $\mathrm{ONB}(\sigma(\mc{H}))$ are \emph{frame functions over product bases}: $f: \sigma(\mc{H}) \rightarrow \mathbb{R}$ with $\sum_{j_1,\cdots,j_n=1}^{d_1,\cdots,d_n} f(v_{j_1,1} \otimes \cdots \otimes v_{j_n,n}) = W$, $d_i := \mathrm{dim}(\mc{H}_i)$ only on product bases, $\beta(\mc{H}) := \{(v_{j_1,1} \otimes \cdots \otimes v_{j_n,n})_{j_1,\cdots,j_n=1}^{d_1,\cdots,d_n} \mid (v_{j_i,i})_{j_i=1}^{d_i} \in \mathrm{ONB}(\mc{H}_i)\}$.
Clearly, $S(\mc{H})$ contains many nonlocal states. But even unentangled bases cannot always be implemented with local operations and classical communication only \cite{WoottersEtAl1999}, suggesting product bases as the most natural choice of constraints. Yet, non-negative frame functions over product bases do not even correspond with linear operators (cf. Prop.~5 in \cite{Wallach2002}).

To gain some insight into what is `missing', it is helpful to consider examples of frame functions over product bases. \cite{Wallach2002} gives a whole family of examples, which are easily seen to correspond to signalling distributions. We thus add more constraints in the form of no-signalling: for $i \in \{1,\cdots,n\}$ with $(v_{j_i,i})_{j_i=1}^{d_i}, (w_{k_i,i})_{k_i=1}^{d_i} \in \mathrm{ONB}(\mc{H}_i)$ and $x_{l_r,r} \in \mc{S}(\mc{H}_r)$ for all $l_r \in \{1,\cdots,d_r\}$, $r \neq i$,
\begin{equation}\label{eq: no-signalling constraints}
    \sum_{j_i=1}^{d_i} f(x_{l_r,r} \otimes v_{j_i,i}) = \sum_{k_i=1}^{d_i} f(x_{l_r,r} \otimes w_{k_i,i})\; ,
\end{equation}
where we use the shorthand $x_{l_r,r} \otimes v_{j_i,i} := (x_{l_1,1} \otimes \cdots \otimes x_{l_{i-1},i-1} \otimes v_{j_i,i} \otimes x_{l_{i+1},i+1} \otimes \cdots \otimes x_{l_n,n})$. In light of PR-boxes one might still expect such \emph{non-signalling frame functions} to be more general than quantum states. Yet, Eq.~(\ref{eq: no-signalling constraints}) is already enough to exlude PR-boxes. To see this, we introduce another choice of basis: let $B \in \beta(\mc{H})$, $B' \in \mathrm{ONB}(\mc{H})$ and set $B' \sim B$ if there exists a sequence of unitaries $(U^m)_{m=1}^N$ such that $B^0 = B$, $B^m = U^m B^{m-1}$, $B^N = B'$ and where every unitary $U^m$ acts non-trivially only on local subspaces of the form $x^m_{l_r,r} \otimes (v^m_{j_i,i} + v^m_{j'_i,i})$ with $x^m_{l_r,r} \otimes v^m_{j_i,i}, x^m_{l_r,r} \otimes v^m_{j'_i,i} \in B^m$. Note that the equivalence relation $\sim$ is independent of the choice of product basis $B \in \beta(\mc{H})$, it only depends on the Hilbert space decomposition $\mc{H} = \otimes_{i=1}^n \mc{H}_i$ (cf. Fig.~\ref{fig: transformations between unentangled and product bases}).
We call the elements in $T(\beta(\mc{H})):= \{B' \in \mathrm{ONB}(H) \mid \exists B \in \beta(\mc{H}):\ B' \sim B \}$ \emph{twisted product bases} \footnote{Twisted product bases arise in a similar (but more complex) way to the set of rotations of a Rubik's cube.}. For instance, the unentangled basis in Fig.~\ref{fig: transformations between unentangled and product bases} (cf. \cite{WoottersEtAl1999}) is easily transformed into a product basis, and is thus also a twisted product basis.

\begin{figure}
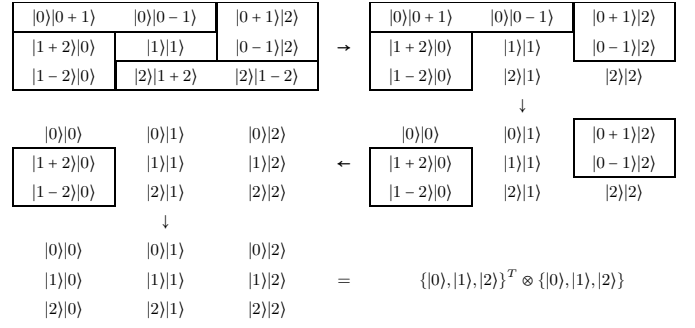

    \scalebox{0.7}{
    {\def\arraystretch{1.5}\tabcolsep=10pt
    \centering
    \begin{tabular}{ccccccc}
        \cline{1-2} \cline{3-3} \cline{5-6} \cline{7-7} \multicolumn{1}{|c}{$|0\rangle|0+1\rangle$} & \multicolumn{1}{c|}{$|0\rangle |0-1\rangle$} & \multicolumn{1}{|c|}{$|0+1\rangle |2\rangle$} & & \multicolumn{1}{|c}{$|0\rangle |0+1\rangle$} & \multicolumn{1}{c|}{$|0\rangle |0-1\rangle$} & \multicolumn{1}{|c|}{$|0+1\rangle |2\rangle$} \\
        \cline{1-1} \cline{1-2} \cline{5-5} \cline{5-6} \multicolumn{1}{|c|}{$|1+2\rangle |0\rangle$} & $|1\rangle |1\rangle$ & \multicolumn{1}{|c|}{$|0-1\rangle|2\rangle$} & $\bm{\rightarrow}$ & \multicolumn{1}{|c|}{$|1+2\rangle |0\rangle$} & $|1\rangle |1\rangle$ & \multicolumn{1}{|c|}{$|0-1\rangle |2\rangle$} \\
        \cline{2-3} \cline{3-3} \cline{7-7} \multicolumn{1}{|c|}{$|1-2\rangle |0\rangle$} & \multicolumn{1}{|c}{$|2\rangle |1+2\rangle$} & \multicolumn{1}{c|}{$|2\rangle |1-2\rangle$} & & \multicolumn{1}{|c|}{$|1-2\rangle |0\rangle$} & $|2\rangle |1\rangle$ & $|2\rangle |2\rangle$ \\
        \cline{1-1} \cline{2-3} \cline{5-5}
        & & & & & $\bm{\downarrow}$ & \\
        \cline{7-7} $|0\rangle |0\rangle$ & $|0\rangle |1\rangle$ & $|0\rangle |2\rangle$ & & $|0\rangle |0\rangle$ & $|0\rangle |1\rangle$ & \multicolumn{1}{|c|}{$|0+1\rangle |2\rangle$} \\
        \cline{1-1} \cline{5-5} \multicolumn{1}{|c|}{$|1+2\rangle |0\rangle$} & $|1\rangle |1\rangle$ & $|1\rangle |2\rangle$ & $\bm{\leftarrow}$ & \multicolumn{1}{|c|}{$|1+2\rangle |0\rangle$} & $|1\rangle |1\rangle$ & \multicolumn{1}{|c|}{$|0-1\rangle |2\rangle$} \\ \cline{7-7} \multicolumn{1}{|c|}{$|1-2\rangle |0\rangle$} & $|2\rangle |1\rangle$ & $|2\rangle |2\rangle$ & & \multicolumn{1}{|c|}{$|1-2\rangle |0\rangle$} & $|2\rangle |1\rangle$ & $|2\rangle |2\rangle$ \\ \cline{1-1} \cline{5-5}
        & $\bm{\downarrow}$ & & & & & \\
        $|0\rangle |0\rangle$ & $|0\rangle |1\rangle$ & $|0\rangle |2\rangle$ & & & & \\
        $|1\rangle |0\rangle$ & $|1\rangle |1\rangle$ & $|1\rangle |2\rangle$  & = & \multicolumn{3}{c}{$\{|0\rangle,|1\rangle,|2\rangle\}^T \otimes \{|0\rangle,|1\rangle,|2\rangle\}$} \\
        $|2\rangle |0\rangle$ & $|2\rangle |1\rangle$ & $|2\rangle |2\rangle$ & & & & 
    \end{tabular}}}
    \caption{The unentangled basis in the top left corner (cf. \cite{WoottersEtAl1999}) is transformed into a product basis (bottom left corner) by successively applying local unitaries, e.g., in the first step $(|2\rangle|1+2\rangle, |2\rangle|1-2\rangle) \rightarrow (|2\rangle|1\rangle, |2\rangle|2\rangle)$ where $|x\rangle|y\rangle := |x\rangle \otimes |y\rangle$ as well as $|x\pm y\rangle := \frac{1}{\sqrt{2}}(|x\rangle \pm |y\rangle)$.}
    \label{fig: transformations between unentangled and product bases}
\end{figure}

We state two important facts about frame functions over twisted product bases. (For proofs and more details we refer to the Supplementary Material.) First, Thm.~2 in \cite{Wallach2002} fails for product bases, yet it \emph{already} holds for frame functions over twisted product bases. Since the latter contain strictly fewer conditions than unentangled frame functions (cf. Prop.~\ref{prop: strict inclusion} in Supp. Mat.), this generalises the result in \cite{Wallach2002}. Second, for frame functions over product bases compatibility with twisting operations is equivalent to no-signalling in Eq.~(\ref{eq: no-signalling constraints}).

\begin{theorem}\label{thm: non-signalling frame functions correspond to self-adjoint operators PVM}
    Let $\mc{H} = \otimes_{i=1}^n \mc{H}_i$ with $\mathrm{dim}(\mc{H}_i) \geq 3$ finite for all $i \in \{1,\cdots,n\}$, $n \in \mathbb{N}$. If $f: \sigma(\mc{H}) \rightarrow \mathbb{R}$ is a non-negative, non-signalling frame function over product bases, then there exists a self-adjoint operator $t: \mc{H} \rightarrow \mc{H}$ such that $f(v) = \mathrm{tr}(t p_v) = \langle v | t | v \rangle$ $\forall v \in \sigma(\mc{H})$, $p_v = |v\rangle\langle v| \in \mc{P}(\mc{H})$.
\end{theorem}


Moreover, \cite{Wehner2010} demonstrates that such non-signalling correlations cannot exceed quantum correlations (cf. App.~\ref{sec: Systems of dimensions two}).
In the next paragraph we reformulate Thm.~\ref{thm: non-signalling frame functions correspond to self-adjoint operators PVM} in contextual form, in particular, no-signalling then arises from marginalisation over product contexts.


\paragraph*{Contextuality and Composition.} At the core of contextuality lies the notion of \emph{simultaneous measurability} \cite{Specker1960}. We say that a system is \emph{contextual} if not all its observables can be measured simultaneously in every state. Clearly, classical systems are non-contextual, but also contextual systems might contain sets of simultaneously measurable observables called \emph{contexts}. Moreover, the set of contexts carries an intrinsic order relation arising from \emph{coarse-graining} on outcomes of observables. The resulting partial order is called the \emph{partial order of contexts}.
While maybe simplistic at first glance, in quantum mechanics this structure captures most aspects of the theory. In fact, the only missing information is a choice of time orientation \cite{Doering2014}. Before introducing the latter, we show how quantum states arise in this picture.

Note first that in quantum theory contexts are mathematically represented by commutative subalgebras denoted $V \subseteq \mc{L}(\mc{H})$, which are ordered by inclusion into the corresponding partial order of contexts denoted by $\mc{V}(\mc{H})$. 
In this setup a quantum state becomes a collection of probability distributions $(\mu_V)_{V \in \mc{V}(\mc{H})}$, one in every context. Moreover, \emph{non-contextuality} constrains these \emph{across} contexts: let $\mu_{\tilde{V}}$, $\mu_V$ be measures in contexts $\tilde{V},V$, $\tilde{V} \subseteq V$, then $\mu_{\tilde{V}}$ is obtained from $\mu_V$ by marginalisation, $\mu_{\tilde{V}} =\mu_V|_{\tilde{V}}$. More formally, we define:

\begin{definition}\label{def: dilated probabilistic presheaf}
    Let $\mc{H}$, $\mc{K}$ be Hilbert spaces and let $\mc{V}(\mc{H})$ be the partial order of contexts over $\mc{H}$. The functor $\PPi: \mc{V}(\mc{H})^\mathrm{op} \rightarrow \mathbf{Set}$,
    \begin{align*}
        \PPi_V &:= \{\mu_V: \mc{P}(V) \rightarrow [0,1] \mid \mu_V = v^* \varphi_V v \text{ for } v \in \mc{K}, \\
        & \hspace{-0.3cm} \varphi_V: \mc{P}(V) \hookrightarrow \mc{P}(\mc{K}) \text{ an embedding, and } \mu_V(1) = 1\}\; , \\
        \PPi(\tilde{V} &\subseteq V): \PPi_V \rightarrow \PPi_{\tilde{V}},\ \mu_V = v^* \varphi_V v \mapsto \mu_{\tilde{V}} = v^* \varphi_V|_{\tilde{V}} v\; ,
    \end{align*}
    is called the \emph{dilated probabilistic presheaf over $\mc{H}$. \footnote{Here, $\mc{K}$ is independent of contexts, since dilations exist as long as $\mathrm{dim}(\mc{K}) \geq \mathrm{dim}(\mc{L}(\mc{H}))$ (by Naimark's theorem).}} 
    
    A global section $\gamma = (\mu^\gamma_V)_{V \in \mc{V}(\mc{H})}$ is a collection of probability distribution, one in every context, such that whenever $\tilde{V},V \in \mc{V}(\mc{H})$, $\tilde{V} \subseteq V$ it follows $\mu^\gamma_V|_{\tilde{V}} = \mu^\gamma_{\tilde{V}}$. The set of all global sections of $\PPi$ is denoted by $\Gamma(\PPi(\mc{V}(\mc{H})))$.
\end{definition}

In words, $\PPi$ assigns to every context $V \in \mc{V}(\mc{H})$ the set of probability distributions over its projection lattice $\mc{P}(V)$, and every inclusion relation between contexts $\tilde{V} \subseteq V$ is represented by the marginalisation map between corresponding sets of probability distributions $\PPi(V) \rightarrow \PPi(\tilde{V})$. (Note that $\PPi$ thus reverses the order on $\mc{V}(\mc{H})$.) We also require that probability distributions admit dilations in contexts. Mathematically, this corresponds to applying Naimark's theorem in contexts \cite{Naimark1943}. In particular, by Gelfand duality we may identify every context $V \in \mc{V}(\mc{H})$ with a compact Hausdorff space such that $\mu_V$ becomes a positive operator-valued measure (of $1 \times 1$-matrices) with spectral dilation $v^*\phi_V v$, where $v \in \mc{K}$ defines a linear map $v: \mathbb{C} \rightarrow \mc{K}$ by scalar multiplication.
Importantly, in the Supplementary Material we show that global sections of $\PPi(\mc{V}(\mc{H}))$ bijectively correspond with quantum states by Gleason's theorem.

In order to extend this to the bipartite case, we further need to consider composition of contexts. The canonical product on partial orders, denoted $\mc{V}_1 \times \mc{V}_2$, is the cartesian product with elements $(V_1,V_2)$ for $V_1 \in \mc{V}_1$, $V_2 \in \mc{V}_2$ and order relations such that for all $\tilde{V}_1,V_1 \in \mc{V}_1$, $\tilde{V}_2,V_2 \in \mc{V}_2$:
\begin{equation}\label{eq: product context category}
    (\tilde{V}_1,\tilde{V}_2) \subseteq (V_1,V_2) :\Longleftrightarrow \tilde{V}_1 \subseteq_1 V_1 \mathrm{\ and \ } \tilde{V}_2 \subseteq_2 V_2\; .
\end{equation}
We define the \emph{Bell presheaf} as the dilated probabilistic presheaf $\PPi(\mc{V}(\mc{H}_1) \times \mc{V}(\mc{H}_2))$ over product contexts. As a consequence of Thm.~\ref{thm: non-signalling frame functions correspond to self-adjoint operators PVM}, we find that no-signalling is contained in the marginalisation maps between product contexts (for more details, see \cite{FreDoe19a}).

\begin{proposition}\label{prop: Bell-type Gleason for global sections}
    Let $\mc{H} = \mc{H}_1 \otimes \mc{H}_2$ with $\mathrm{dim}(\mc{H}_i) \geq 3$ finite. For every global section of the Bell presheaf $\gamma \in \Gamma(\PPi(\mc{V}(\mc{H}_1) \times \mc{V}(\mc{H}_2)))$ there exists a self-adjoint operator $t: \mc{H} \rightarrow \mc{H}$ of unit trace such that $\gamma(p \otimes q) = \mathrm{tr}(t (p \otimes q)) \geq 0$ for all $p \in \mc{P}(\mc{H}_1)$, $q \in \mc{P}(\mc{H}_2)$.
\end{proposition}

Prop.~\ref{prop: Bell-type Gleason for global sections}
shows that no-signalling \emph{almost} restricts global sections of the Bell presheaf $\PPi(\mc{V}(\mc{H}_1) \times \mc{V}(\mc{H}_2))$
to density matrices. In fact, we will soon strengthen Prop.~\ref{prop: Bell-type Gleason for global sections} to a bijective correspondence via Lm.~\ref{lm: global sections to Jordan homos} below. In order to do so, we first need to consider Jordan algebras.

\paragraph*{Jordan Algebras.} Let $\mc{J}(\mc{H}) = (\mc{L}_\mathrm{sa}(\mc{H}),\{\cdot,\cdot\})$ denote the \emph{Jordan algebra} corresponding to $\mc{H}$, i.e., the set of self-adjoint matrices $\mc{L}_\mathrm{sa}(\mc{H}) := \{a \in \mc{L}(\mc{H}) \mid a^* = a\}$ with product given by the \emph{anticommutator} $\{a,b\} = ab + ba$ for all $a,b \in \mc{L}_\mathrm{sa}(\mc{H})$. $\mc{J}(\mc{H})$ extends to the complexification $\mc{L}_\mathrm{sa}(\mc{H}) + i\mc{L}_\mathrm{sa}(\mc{H})$ and we will denote the complexified algebra by $\mc{J}(\mc{H})$ also. A Jordan ($*$-)~homomorphism $\Phi: \mc{J}_1 \rightarrow \mc{J}_2$ is a linear map such that $\Phi(\{a,b\}) = \{\Phi(a),\Phi(b)\}$ (and $\Phi(a^*) = \Phi(a)^*$) for all $a,b \in \mc{J}_1$. Clearly, $\mc{J}(\mc{H})$ contains less information than $\mc{L}(\mc{H})$, it misses the antisymmetric part, i.e., the \emph{commutator} $[a,b] = ab - ba$, in the associative product $ab = \frac{1}{2}(ab + ba) + \frac{1}{2}(ab - ba)$ for $a,b \in \mc{L}(\mc{H})$. Conversely, it turns out that commutators can be added to $\mc{J}(\mc{H})$ in only two ways (cf. \cite{Kadison1951,AlfsenShultz1998}), which are distinguished by a sign choice: $\frac{1}{2}(ab + ba) \pm \frac{1}{2}(ab - ba)$. We denote the corresponding associative algebras by
\begin{align*}
    \mc{L}(\mc{H}) &:= \{a \in \mc{J}(\mc{H}) \mid a \circ b = \frac{1}{2}\{a,b\} + \frac{1}{2}[a,b] = ab\}\; , \\
    \overline{\mc{L}(\mc{H})} &:= \{a \in \mc{J}(\mc{H}) \mid a \circ b = \frac{1}{2}\{a,b\} - \frac{1}{2}[a,b] = ba\}\; .
\end{align*}

Clearly, $\mc{V}(\mc{L}(\mc{H})) = \mc{V}(\overline{\mc{L}(\mc{H})}) = \mc{V}(\mc{H})$. In fact, $\mc{V}(\mc{H})$ already determines the Jordan algebra structure of $\mc{L}(\mc{H})$ \cite{DoeringHarding2016}. This allows us to prove a refinement of Prop.~\ref{thm: non-signalling frame functions correspond to self-adjoint operators PVM}.

\begin{lemma}\label{lm: global sections to Jordan homos}
    Let $\mc{H} = \mc{H}_1 \otimes \mc{H}_2$ with $\mathrm{dim}(\mc{H}_i) \geq 3$ finite. For every global section of the Bell presheaf $\gamma \in \Gamma(\PPi(\mc{V}(\mc{H}_1) \times \mc{V}(\mc{H}_2)))$ there exists a unique linear map $\phi^\gamma: \mc{L}(\mc{H}_1) \rightarrow \mc{L}(\mc{H}_2)$, moreover, a Hilbert space $\mc{K}$, a linear map $v: \mc{H}_2 \rightarrow \mc{K}$, and a Jordan $*$-homomorphism $\Phi^\gamma: \mc{J}(\mc{H}_1) \rightarrow \mc{J}(\mc{K})$ such that $\phi^\gamma = v^* \Phi^\gamma v$.
\end{lemma}

Now note that by Stinespring's theorem, a linear map $\phi: \mc{L}(\mc{H}_1) \rightarrow \mc{L}(\mc{H}_2)$ is completely positive if and only if it is of a similar form, namely $\phi = v^* \Phi v$ with $\Phi: \mc{L}(\mc{H}_1) \rightarrow \mc{L}(\mc{K})$ a $*$-homomorphism, i.e., $\Phi(ab) = \Phi(a)\Phi(b)$ and $\Phi(a^*) = \Phi(a)^*$ for all $a,b \in \mc{L}(\mc{H}_1)$. Furthermore, by Choi's theorem, trace-preserving completely positive maps correspond to quantum states. More precisely, the matrix $(\rho_\phi)_{ij} = \phi(E_{ij})$ arising from the Choi-Jamio\l kowski isomorphism applied to $\phi$, where $E_{ij} \in \mc{L}(\mc{H}_1)$ denotes the matrix with entry $1$ in position $(i,j)$ and $0$ otherwise, is positive if and only if $\phi$ is completely positive. By the correspondence between states and density matrices in finite dimensions,
$\rho_\phi$ therefore corresponds to the state $\sigma_\phi := \mathrm{tr}(\rho_\phi \cdot \underline{\ \ }) \in \mc{S}(\mc{L}(\mc{H}_1) \otimes \mc{L}(\mc{H}_2))$.

Clearly, $\phi^\gamma = v^* \Phi^\gamma v: \mc{L}(\mc{H}_1) \rightarrow \mc{L}(\mc{H}_2)$ in Lm.~\ref{lm: global sections to Jordan homos} is positive, however, it is generally not completely positive. Consequently, no-signalling in Thm.~\ref{thm: non-signalling frame functions correspond to self-adjoint operators PVM}, equivalently, the mere order structure on product contexts, is not enough to single out quantum states on $\mc{L}(\mc{H}) = \mc{L}(\mc{H}_1) \otimes \mc{L}(\mc{H}_2)$.
Nevertheless, Lm.~\ref{lm: global sections to Jordan homos} highlights just what is missing for global sections to bijectively correspond to quantum states. Namely, $\Phi^\gamma$ needs to lift from a Jordan $*$-homomorphism to a $*$-homomorphism. Since the former already preserves anticommutators, this is equivalent to imposing a consistency condition between commutators in $\mc{L}(\mc{H}_1)$ and $\mc{L}(\mc{H}_2)$. Finally, since commutators generate infinitesimal time evolution, and Jordan $*$-homomorphisms already fix commutators up to sign (cf. \cite{AlfsenShultz1998}), this extra information can be interpreted as a choice of local time orientation.

\paragraph*{Local Time Orientations.} We define a \emph{time orientation on $\mc{V}(\mc{H})$} as a map $\psi: \mathbb{R} \times \mc{L}_\mathrm{sa}(\mc{H}) \rightarrow \mathrm{Aut}(\mc{V}(\mc{H}))$,
where $t \mapsto \psi(t,a)$ defines a one-parameter group of order automorphisms for every $a \in \mc{L}_\mathrm{sa}(\mc{H})$.
By Wigner's theorem, every $\psi(t,a) \in \mathrm{Aut}(\mc{V}(\mc{H}))$
is given by
conjugation with a unitary or anti-unitary operator, $\psi(t,a)(V) = U(t,a)VU(t,a)^*$.
Moreover, since every anti-unitary operator is the product of a unitary operator $e^{ita}$ and the time reversal operator, we may interpret the parameter $t$ as time and observe that $\psi$ fixes the forward time direction of the system described by $\mc{V}(\mc{H})$.
In particular, given the associative algebra $\mc{L}(\mc{H})$, we may define a \emph{canonical time orientation} by $\psi(t,a)(V) = e^{ita}Ve^{-ita}$, which relates to the commutator in $\mc{L}(\mc{H})$ via $\frac{d}{dt}|_{t=0} (e^{ita}be^{-ita}) = i[a,b]$.

We define the \emph{time-oriented partial order of contexts} $\widetilde{\mc{V}(\mc{H})}$ as $\mc{V}(\mc{H})$ together with a time orientation $\psi$. Clearly, for the product context order $\mc{V}(\mc{H}_1) \times \mc{V}(\mc{H}_2)$ orientations exist on each subsystem independently, $\psi = (\psi_1,\psi_2)$. Finally, we say a global section 
$\gamma \in \Gamma(\PPi(\mc{V}(\mc{H}_1) \times \mc{V}(\mc{H}_2)))$ is \emph{orientation-preserving} if it preserves this time orientation, i.e.,
\begin{equation*}\label{eq: Jordan to vN condition}
    \forall t \in \mathbb{R}, a \in \mc{L}_\mathrm{sa}(\mc{H}_1):\ \Phi^\gamma \circ \psi_1(t,a) = \psi_2(t,\Phi^\gamma(a)) \circ \Phi^\gamma\; , 
\end{equation*}
where $\Phi^\gamma$ is the Jordan $*$-homomorphism in Lm.~\ref{lm: global sections to Jordan homos}.
For more details on time orientations we refer to \cite{Doering2014} (and Supp. Mat.). With these definitions in place,
we are ready to lift Prop.~\ref{prop: Bell-type Gleason for global sections} to a bijective correspondence.

\vspace{-0.1cm}
\begin{theorem}\label{thm: Bell-type Gleason}
    Let $\mc{H} = \mc{H}_1 \otimes \mc{H}_2$ with $\mathrm{dim}(\mc{H}_i) \geq 3$ finite. There is a bijective correspondence between the set of orientation-preserving global sections of the Bell presheaf $\Gamma(\PPi(\widetilde{\mc{V}(\mc{H}_1)} \times \widetilde{\mc{V}(\mc{H}_2)}))$ and (quantum) states $\mc{S}(\mc{L}(\mc{H})) = \mc{S}(\mc{L}(\mc{H}_1) \otimes \mc{L}(\mc{H}_2))$.
\end{theorem}
\vspace{-0.1cm}

In fact, every global section of the Bell presheaf, $\gamma \in \Gamma(\PPi(\mc{V}(\mc{H}_1) \times \mc{V}(\mc{H}_2)))$, corresponds to a quantum state since by the classification in \cite{Kadison1951,AlfsenShultz1998} we can always lift $\Phi^\gamma$ to a $*$-homomorphism by choosing orientations on $\mc{V}(\mc{H}_1)$ and $\mc{V}(\mc{H}_2)$ appropriately.
This improves a previous result in \cite{KlayRandallFoulis1987,Wehner2010}: locally quantum non-signalling correlations always admit a quantum description.



\paragraph*{Conclusion.} We discussed the physical principle of no-signalling and its relation to contextuality.
Our main theorem, Thm.~\ref{thm: Bell-type Gleason}, is a generalisation of Gleason's theorem over product contexts, which complements earlier results in \cite{KlayRandallFoulis1987,Wallach2002,Wehner2010}. In particular, we related non-signalling correlations with quantum states unambiguously. Succinctly: (i) no-signalling arises via marginalisation constraints on probability distributions in product contexts, (ii)
complete positivity reduces to marginalisation for dilations in product contexts and a consistency condition between local time orientations,
and (iii)
non-signalling dilations over product contexts bijectively correspond with quantum states for appropriate time orientations.



\paragraph*{Acknowledgements.} We thank Vicky Wright for pointing us to \cite{Wallach2002} 
and Nolan Wallach for helpful discussions. This work is supported through a studentship in the Centre for Doctoral Training on Controlled Quantum Dynamics at Imperial College funded by the EPSRC.

\bibliographystyle{apsrev4-1}
\bibliography{bibliography}

\begin{thebibliography}{31}%
\makeatletter
\providecommand \@ifxundefined [1]{%
 \@ifx{#1\undefined}
}%
\providecommand \@ifnum [1]{%
 \ifnum #1\expandafter \@firstoftwo
 \else \expandafter \@secondoftwo
 \fi
}%
\providecommand \@ifx [1]{%
 \ifx #1\expandafter \@firstoftwo
 \else \expandafter \@secondoftwo
 \fi
}%
\providecommand \natexlab [1]{#1}%
\providecommand \enquote  [1]{``#1''}%
\providecommand \bibnamefont  [1]{#1}%
\providecommand \bibfnamefont [1]{#1}%
\providecommand \citenamefont [1]{#1}%
\providecommand \href@noop [0]{\@secondoftwo}%
\providecommand \href [0]{\begingroup \@sanitize@url \@href}%
\providecommand \@href[1]{\@@startlink{#1}\@@href}%
\providecommand \@@href[1]{\endgroup#1\@@endlink}%
\providecommand \@sanitize@url [0]{\catcode `\\12\catcode `\$12\catcode
  `\&12\catcode `\#12\catcode `\^12\catcode `\_12\catcode `\%12\relax}%
\providecommand \@@startlink[1]{}%
\providecommand \@@endlink[0]{}%
\providecommand \url  [0]{\begingroup\@sanitize@url \@url }%
\providecommand \@url [1]{\endgroup\@href {#1}{\urlprefix }}%
\providecommand \urlprefix  [0]{URL }%
\providecommand \Eprint [0]{\href }%
\providecommand \doibase [0]{http://dx.doi.org/}%
\providecommand \selectlanguage [0]{\@gobble}%
\providecommand \bibinfo  [0]{\@secondoftwo}%
\providecommand \bibfield  [0]{\@secondoftwo}%
\providecommand \translation [1]{[#1]}%
\providecommand \BibitemOpen [0]{}%
\providecommand \bibitemStop [0]{}%
\providecommand \bibitemNoStop [0]{.\EOS\space}%
\providecommand \EOS [0]{\spacefactor3000\relax}%
\providecommand \BibitemShut  [1]{\csname bibitem#1\endcsname}%
\let\auto@bib@innerbib\@empty
\bibitem [{\citenamefont {Clauser}\ \emph {et~al.}(1969)\citenamefont {Clauser}
  \emph {et~al.}}]{CHSH}%
  \BibitemOpen
  \bibfield  {author} {\bibinfo {author} {\bibfnamefont {J.~F.}\ \bibnamefont
  {Clauser}} \emph {et~al.},\ }\href@noop {} {\bibfield  {journal} {\bibinfo
  {journal} {Phys. Rev. Lett.}\ }\textbf {\bibinfo {volume} {23}},\ \bibinfo
  {pages} {880} (\bibinfo {year} {1969})}\BibitemShut {NoStop}%
\bibitem [{\citenamefont {{Fine}}(1982)}]{Fine1982}%
  \BibitemOpen
  \bibfield  {author} {\bibinfo {author} {\bibfnamefont {A.}~\bibnamefont
  {{Fine}}},\ }\href@noop {} {\bibfield  {journal} {\bibinfo  {journal} {Phys.
  Rev. Lett.}\ }\textbf {\bibinfo {volume} {48}},\ \bibinfo {pages} {291}
  (\bibinfo {year} {1982})}\BibitemShut {NoStop}%
\bibitem [{\citenamefont {Tsirelson}(1980)}]{Tsirelson1980}%
  \BibitemOpen
  \bibfield  {author} {\bibinfo {author} {\bibfnamefont {B.~S.}\ \bibnamefont
  {Tsirelson}},\ }\href@noop {} {\bibfield  {journal} {\bibinfo  {journal}
  {Lett. Math. Phys.}\ }\textbf {\bibinfo {volume} {4}},\ \bibinfo {pages} {93}
  (\bibinfo {year} {1980})}\BibitemShut {NoStop}%
\bibitem [{\citenamefont {Shalm}\ \emph {et~al.}(2015)\citenamefont {Shalm}
  \emph {et~al.}}]{ShalmEtAl2015}%
  \BibitemOpen
  \bibfield  {author} {\bibinfo {author} {\bibfnamefont {L.~K.}\ \bibnamefont
  {Shalm}} \emph {et~al.},\ }\href@noop {} {\bibfield  {journal} {\bibinfo
  {journal} {Phys. Rev. Lett.}\ }\textbf {\bibinfo {volume} {115}},\ \bibinfo
  {pages} {250402} (\bibinfo {year} {2015})}\BibitemShut {NoStop}%
\bibitem [{\citenamefont {Giustina}\ \emph {et~al.}(2015)\citenamefont
  {Giustina} \emph {et~al.}}]{ZeilingerEtAl2015}%
  \BibitemOpen
  \bibfield  {author} {\bibinfo {author} {\bibfnamefont {M.}~\bibnamefont
  {Giustina}} \emph {et~al.},\ }\href@noop {} {\bibfield  {journal} {\bibinfo
  {journal} {Phys. Rev. Lett.}\ }\textbf {\bibinfo {volume} {115}},\ \bibinfo
  {pages} {250401} (\bibinfo {year} {2015})}\BibitemShut {NoStop}%
\bibitem [{\citenamefont {Popescu}\ and\ \citenamefont
  {Rohrlich}(1994)}]{PopescuRohrlich1994}%
  \BibitemOpen
  \bibfield  {author} {\bibinfo {author} {\bibfnamefont {S.}~\bibnamefont
  {Popescu}}\ and\ \bibinfo {author} {\bibfnamefont {D.}~\bibnamefont
  {Rohrlich}},\ }\href@noop {} {\bibfield  {journal} {\bibinfo  {journal}
  {Found. Phys.}\ }\textbf {\bibinfo {volume} {24}},\ \bibinfo {pages} {379}
  (\bibinfo {year} {1994})}\BibitemShut {NoStop}%
\bibitem [{\citenamefont {{Paw{\l}owski}}\ \emph {et~al.}(2009)\citenamefont
  {{Paw{\l}owski}} \emph {et~al.}}]{PawlowskiEtAl2009}%
  \BibitemOpen
  \bibfield  {author} {\bibinfo {author} {\bibfnamefont {M.}~\bibnamefont
  {{Paw{\l}owski}}} \emph {et~al.},\ }\href@noop {} {\bibfield  {journal}
  {\bibinfo  {journal} {\nat}\ }\textbf {\bibinfo {volume} {461}},\ \bibinfo
  {pages} {1101} (\bibinfo {year} {2009})}\BibitemShut {NoStop}%
\bibitem [{\citenamefont {van Dam}(2013)}]{vanDam2013}%
  \BibitemOpen
  \bibfield  {author} {\bibinfo {author} {\bibfnamefont {W.}~\bibnamefont {van
  Dam}},\ }\href@noop {} {\bibfield  {journal} {\bibinfo  {journal} {Natural
  Computing}\ }\textbf {\bibinfo {volume} {12}},\ \bibinfo {pages} {9}
  (\bibinfo {year} {2013})}\BibitemShut {NoStop}%
\bibitem [{\citenamefont {{Barnum}}\ \emph {et~al.}(2010)\citenamefont
  {{Barnum}} \emph {et~al.}}]{Wehner2010}%
  \BibitemOpen
  \bibfield  {author} {\bibinfo {author} {\bibfnamefont {H.}~\bibnamefont
  {{Barnum}}} \emph {et~al.},\ }\href@noop {} {\bibfield  {journal} {\bibinfo
  {journal} {Phys. Rev. Lett.}\ }\textbf {\bibinfo {volume} {104}} (\bibinfo
  {year} {2010})}\BibitemShut {NoStop}%
\bibitem [{\citenamefont {{Kl{\"a}y}}\ \emph {et~al.}(1987)\citenamefont
  {{Kl{\"a}y}} \emph {et~al.}}]{KlayRandallFoulis1987}%
  \BibitemOpen
  \bibfield  {author} {\bibinfo {author} {\bibfnamefont {M.}~\bibnamefont
  {{Kl{\"a}y}}} \emph {et~al.},\ }\href@noop {} {\bibfield  {journal} {\bibinfo
   {journal} {Int. J. Theor. Phys.}\ }\textbf {\bibinfo {volume} {26}},\
  \bibinfo {pages} {199} (\bibinfo {year} {1987})}\BibitemShut {NoStop}%
\bibitem [{\citenamefont {{Wallach}}(2002)}]{Wallach2002}%
  \BibitemOpen
  \bibfield  {author} {\bibinfo {author} {\bibfnamefont {N.~R.}\ \bibnamefont
  {{Wallach}}},\ }\enquote {\bibinfo {title} {An unentangled gleason's
  theorem},}\ in\ \href@noop {} {\emph {\bibinfo {booktitle} {Quantum
  computation and information}}},\ Vol.\ \bibinfo {volume} {305}\ (\bibinfo
  {publisher} {Amer. Math. Soc.},\ \bibinfo {year} {2002})\ p.\ \bibinfo
  {pages} {291–298}\BibitemShut {NoStop}%
\bibitem [{\citenamefont {Kochen}\ and\ \citenamefont
  {Specker}(1967)}]{KochenSpecker1967}%
  \BibitemOpen
  \bibfield  {author} {\bibinfo {author} {\bibfnamefont {S.}~\bibnamefont
  {Kochen}}\ and\ \bibinfo {author} {\bibfnamefont {E.~P.}\ \bibnamefont
  {Specker}},\ }\href@noop {} {\bibfield  {journal} {\bibinfo  {journal} {J.
  Math. Mech.}\ }\textbf {\bibinfo {volume} {17}},\ \bibinfo {pages} {59}
  (\bibinfo {year} {1967})}\BibitemShut {NoStop}%
\bibitem [{\citenamefont {D{\"o}ring}\ and\ \citenamefont
  {Frembs}(2019)}]{FreDoe19a}%
  \BibitemOpen
  \bibfield  {author} {\bibinfo {author} {\bibfnamefont {A.}~\bibnamefont
  {D{\"o}ring}}\ and\ \bibinfo {author} {\bibfnamefont {M.}~\bibnamefont
  {Frembs}},\ }\href@noop {} {\bibfield  {journal} {\bibinfo  {journal} {ArXiv
  e-prints}\ } (\bibinfo {year} {2019})},\ \Eprint
  {http://arxiv.org/abs/1910.09591} {arXiv:1910.09591} \BibitemShut {NoStop}%
\bibitem [{\citenamefont {Gleason}(1957)}]{Gleason1975}%
  \BibitemOpen
  \bibfield  {author} {\bibinfo {author} {\bibfnamefont {A.}~\bibnamefont
  {Gleason}},\ }\href@noop {} {\bibfield  {journal} {\bibinfo  {journal} {J.
  Math. Mech.}\ }\textbf {\bibinfo {volume} {6}},\ \bibinfo {pages} {885}
  (\bibinfo {year} {1957})}\BibitemShut {NoStop}%
\bibitem [{\citenamefont {{Bennett}}\ \emph {et~al.}(1999)\citenamefont
  {{Bennett}} \emph {et~al.}}]{WoottersEtAl1999}%
  \BibitemOpen
  \bibfield  {author} {\bibinfo {author} {\bibfnamefont {C.~H.}\ \bibnamefont
  {{Bennett}}} \emph {et~al.},\ }\href@noop {} {\bibfield  {journal} {\bibinfo
  {journal} {\pra}\ }\textbf {\bibinfo {volume} {59}},\ \bibinfo {pages} {1070}
  (\bibinfo {year} {1999})}\BibitemShut {NoStop}%
\bibitem [{Note1()}]{Note1}%
  \BibitemOpen
  \bibinfo {note} {Twisted product bases arise in a similar (but more complex)
  way to the set of rotations of a Rubik's cube.}\BibitemShut {Stop}%
\bibitem [{\citenamefont {Specker}(1960)}]{Specker1960}%
  \BibitemOpen
  \bibfield  {author} {\bibinfo {author} {\bibfnamefont {E.}~\bibnamefont
  {Specker}},\ }\href@noop {} {\bibfield  {journal} {\bibinfo  {journal}
  {Dialectica}\ }\textbf {\bibinfo {volume} {14}},\ \bibinfo {pages} {239}
  (\bibinfo {year} {1960})}\BibitemShut {NoStop}%
\bibitem [{\citenamefont {D{\"o}ring}(2014)}]{Doering2014}%
  \BibitemOpen
  \bibfield  {author} {\bibinfo {author} {\bibfnamefont {A.}~\bibnamefont
  {D{\"o}ring}},\ }\href@noop {} {\bibfield  {journal} {\bibinfo  {journal}
  {ArXiv e-prints}\ } (\bibinfo {year} {2014})},\ \Eprint
  {http://arxiv.org/abs/1411.5558} {arXiv:1411.5558} \BibitemShut {NoStop}%
\bibitem [{Note2()}]{Note2}%
  \BibitemOpen
  \bibinfo {note} {Here, $\protect \mathcal {K}$ is independent of contexts,
  since dilations exist as long as $\protect \mathrm {dim}(\protect \mathcal
  {K}) \geq \protect \mathrm {dim}(\protect \mathcal {L}(\protect \mathcal
  {H}))$ (by Naimark's theorem).}\BibitemShut {Stop}%
\bibitem [{\citenamefont {{Naimark}}(1943)}]{Naimark1943}%
  \BibitemOpen
  \bibfield  {author} {\bibinfo {author} {\bibfnamefont {M.~A.}\ \bibnamefont
  {{Naimark}}},\ }\href@noop {} {\bibfield  {journal} {\bibinfo  {journal} {{C.
  R. (Dokl.) Acad. Sci. URSS, n. Ser.}}\ }\textbf {\bibinfo {volume} {41}},\
  \bibinfo {pages} {359} (\bibinfo {year} {1943})}\BibitemShut {NoStop}%
\bibitem [{\citenamefont {Kadison}(1951)}]{Kadison1951}%
  \BibitemOpen
  \bibfield  {author} {\bibinfo {author} {\bibfnamefont {R.~V.}\ \bibnamefont
  {Kadison}},\ }\href@noop {} {\bibfield  {journal} {\bibinfo  {journal} {Ann.
  Math.}\ }\textbf {\bibinfo {volume} {54}},\ \bibinfo {pages} {325} (\bibinfo
  {year} {1951})}\BibitemShut {NoStop}%
\bibitem [{\citenamefont {{Alfsen}}\ and\ \citenamefont
  {{Shultz}}(1998)}]{AlfsenShultz1998}%
  \BibitemOpen
  \bibfield  {author} {\bibinfo {author} {\bibfnamefont {E.~M.}\ \bibnamefont
  {{Alfsen}}}\ and\ \bibinfo {author} {\bibfnamefont {F.~W.}\ \bibnamefont
  {{Shultz}}},\ }\href@noop {} {\bibfield  {journal} {\bibinfo  {journal}
  {Proc. Natl. Acad. Sci. U.S.A.}\ }\textbf {\bibinfo {volume} {95}},\ \bibinfo
  {pages} {6596} (\bibinfo {year} {1998})}\BibitemShut {NoStop}%
\bibitem [{\citenamefont {D{\"o}ring}\ and\ \citenamefont
  {Harding}(2016)}]{DoeringHarding2016}%
  \BibitemOpen
  \bibfield  {author} {\bibinfo {author} {\bibfnamefont {A.}~\bibnamefont
  {D{\"o}ring}}\ and\ \bibinfo {author} {\bibfnamefont {J.}~\bibnamefont
  {Harding}},\ }\href@noop {} {\bibfield  {journal} {\bibinfo  {journal}
  {Houston J. Math.}\ }\textbf {\bibinfo {volume} {42}},\ \bibinfo {pages}
  {559} (\bibinfo {year} {2016})}\BibitemShut {NoStop}%
\bibitem [{\citenamefont {Lagarias}\ and\ \citenamefont
  {W.~Shor}(1996)}]{LagariasShor1996}%
  \BibitemOpen
  \bibfield  {author} {\bibinfo {author} {\bibfnamefont {J.}~\bibnamefont
  {Lagarias}}\ and\ \bibinfo {author} {\bibfnamefont {P.}~\bibnamefont
  {W.~Shor}},\ }\href@noop {} {\bibfield  {journal} {\bibinfo  {journal} {Bull.
  Amer. Math. Soc.}\ }\textbf {\bibinfo {volume} {27}} (\bibinfo {year}
  {1996})}\BibitemShut {NoStop}%
\bibitem [{\citenamefont {Busch}(2003)}]{Busch2003}%
  \BibitemOpen
  \bibfield  {author} {\bibinfo {author} {\bibfnamefont {P.}~\bibnamefont
  {Busch}},\ }\href@noop {} {\bibfield  {journal} {\bibinfo  {journal} {Phys.
  Rev. Lett.}\ }\textbf {\bibinfo {volume} {91}},\ \bibinfo {pages} {120403}
  (\bibinfo {year} {2003})}\BibitemShut {NoStop}%
\bibitem [{\citenamefont {{Caves}}\ \emph {et~al.}(2004)\citenamefont {{Caves}}
  \emph {et~al.}}]{CavesFuchs2004}%
  \BibitemOpen
  \bibfield  {author} {\bibinfo {author} {\bibfnamefont {C.~M.}\ \bibnamefont
  {{Caves}}} \emph {et~al.},\ }\href@noop {} {\bibfield  {journal} {\bibinfo
  {journal} {Found. Phys.}\ }\textbf {\bibinfo {volume} {34}},\ \bibinfo
  {pages} {193} (\bibinfo {year} {2004})}\BibitemShut {NoStop}%
\bibitem [{\citenamefont {Wright}\ and\ \citenamefont
  {Weigert}(2018)}]{WrightWeigert2018}%
  \BibitemOpen
  \bibfield  {author} {\bibinfo {author} {\bibfnamefont {V.}~\bibnamefont
  {Wright}}\ and\ \bibinfo {author} {\bibfnamefont {S.}~\bibnamefont
  {Weigert}},\ }\href@noop {} {\bibfield  {journal} {\bibinfo  {journal} {J.
  Phys. A}\ }\textbf {\bibinfo {volume} {52}} (\bibinfo {year}
  {2018})}\BibitemShut {NoStop}%
\bibitem [{\citenamefont {{D{\"o}ring}}(2005)}]{Doering2005}%
  \BibitemOpen
  \bibfield  {author} {\bibinfo {author} {\bibfnamefont {A.}~\bibnamefont
  {{D{\"o}ring}}},\ }\href@noop {} {\bibfield  {journal} {\bibinfo  {journal}
  {Int. J. Theor. Phys.}\ }\textbf {\bibinfo {volume} {44}},\ \bibinfo {pages}
  {139} (\bibinfo {year} {2005})}\BibitemShut {NoStop}%
\bibitem [{\citenamefont {D{\"o}ring}(2012)}]{Doering2012b}%
  \BibitemOpen
  \bibfield  {author} {\bibinfo {author} {\bibfnamefont {A.}~\bibnamefont
  {D{\"o}ring}},\ }\href@noop {} {\bibfield  {journal} {\bibinfo  {journal}
  {ArXiv e-prints}\ } (\bibinfo {year} {2012})},\ \Eprint
  {http://arxiv.org/abs/1212.4882} {arXiv:1212.4882} \BibitemShut {NoStop}%
\bibitem [{\citenamefont {J.~Bunce}\ and\ \citenamefont
  {Wright}(1993)}]{BunceWright1993}%
  \BibitemOpen
  \bibfield  {author} {\bibinfo {author} {\bibfnamefont {L.}~\bibnamefont
  {J.~Bunce}}\ and\ \bibinfo {author} {\bibfnamefont {J.}~\bibnamefont
  {Wright}},\ }\href@noop {} {\bibfield  {journal} {\bibinfo  {journal} {Expo.
  Math}\ }\textbf {\bibinfo {volume} {11}},\ \bibinfo {pages} {91} (\bibinfo
  {year} {1993})}\BibitemShut {NoStop}%
\bibitem [{\citenamefont {J.~Bunce}\ and\ \citenamefont
  {Wright}(1992)}]{BunceWright1992}%
  \BibitemOpen
  \bibfield  {author} {\bibinfo {author} {\bibfnamefont {L.}~\bibnamefont
  {J.~Bunce}}\ and\ \bibinfo {author} {\bibfnamefont {J.}~\bibnamefont
  {Wright}},\ }\href@noop {} {\bibfield  {journal} {\bibinfo  {journal} {Bull.
  Amer. Math. Soc.}\ }\textbf {\bibinfo {volume} {26}} (\bibinfo {year}
  {1992})}\BibitemShut {NoStop}%
\end{thebibliography}%


\appendix

\section{Proof of Theorem \ref{thm: non-signalling frame functions correspond to self-adjoint operators PVM}}\label{sec: Proof of Theorem 1}

In this section we provide some background on frame functions on composite systems. By generalising the concept of (non-negative) unentangled frame functions in \cite{Wallach2002} to (non-negative) frame functions over twisted product bases, this will allow us to relate with locally quantum, non-signalling probability distributions as introduced in \cite{KlayRandallFoulis1987,Wehner2010}, which is the content of Thm.~\ref{thm: non-signalling frame functions correspond to self-adjoint operators PVM}.

To this end, we represent local observables by self-adjoint operators $a \in \mc{L}_\mathrm{sa}(\mc{H})$, in particular, we take their outcomes to correspond to closed subspaces in the Hilbert space $\mc{H}$ spanned by basis vectors in orthonormal bases $v \in \mathrm{ONB}(\mc{H})$ or, equivalently, sets of orthogonal projections $p,p' \in \mc{P}(\mc{H})$, $pp'=0$. In taking the set of all outcomes to correspond to projections in Hilbert space, the following question arises: What measures exist on this set? A \emph{measure} in this setting is a map $\mu: \mc{P}(\mc{H}) \rightarrow [0,1]$ such that $\mu(p+p') = \mu(p) + \mu(p')$ whenever $p,p' \in \mc{P}(\mc{H})$, $pp'=0$ and $\mu(1) = 1$. Proposed by Mackey, it prompted Gleason to prove the following powerful result \cite{Gleason1975}.

\begin{theorem}\label{thm: Gleason theorem}
    \textbf{\emph{(Gleason \cite{Gleason1975})}} Let $\mu$ be a measure on the projections of a (real or complex) Hilbert space $\mc{H}$ of finite dimension $\mathrm{dim}(\mc{H}) \geq 3$. Then there exists a density matrix $\rho: \mc{H} \rightarrow \mc{H}$ such that for all projections $p \in \mc{P}(\mc{H})$,
    \begin{equation*}
        \mu(p) = \mathrm{tr}(\rho p)\; .
    \end{equation*}
\end{theorem}

A closely related concept is that of \emph{frame functions of weight $W \in \mathbb{R}$} on the unit sphere $S(\mc{H})$, $f: S(\mc{H}) \rightarrow \mathbb{R}$, where $\sum_{j=1}^d f(v_j) = W$ for all orthonormal bases $(v_j)_{j=1}^d \in \mathrm{ONB}(\mc{H})$ with $d := \mathrm{dim}(\mc{H})$. In fact, Thm.~\ref{thm: Gleason theorem} is a consequence of the following theorem about frame functions.

\begin{theorem}\label{thm: Gleason theorem for frame functions}
    \textbf{\emph{(Gleason for frame functions \cite{Gleason1975})}} Let $\mathrm{dim}(\mc{H}) \geq 3$ be finite. If $f$ is a non-negative frame function of weight $W \in \mathbb{R}^+$, then there exists a density matrix $\rho: \mc{H} \rightarrow \mc{H}$ such that $f(v) = W \mathrm{tr}(\rho p_v) = W \langle v| \rho |v \rangle$ for all $v \in S(\mc{H})$, $p_v = |v\rangle\langle v| \in \mc{P}(\mc{H})$.
\end{theorem}

Of course, we can apply Thm.~\ref{thm: Gleason theorem} to composite quantum systems and consider frame functions $f: S(\mc{H}) \rightarrow \mathbb{R}$, where $\mc{H} = \otimes_{i=1}^n \mc{H}_i$ is the tensor product Hilbert space. However, in doing so we no longer restrict to outcomes of local measurements only. Consequently, \cite{Wallach2002} restricts to unentangled frame functions $f: \sigma(\mc{H}) \rightarrow \mathbb{R}$ with constraints arising from bases consisting of product states only, i.e., $\sum_{j=1}^d f(v_j) = W$, $d := \mathrm{dim}(\mc{H})$ with $(v_j)_{j=1}^d \in \mathrm{ONB}(\sigma(\mc{H}))$, $\sigma(\mc{H}) := \{v_1 \otimes \cdots \otimes v_n \in S(\mc{H}) \mid v_i \in S(\mc{H}_i)\}$. Under this restriction, \cite{Wallach2002} derives the following result.



\begin{theorem}\label{thm: Wallach}
    \textbf{\emph{(Wallach \cite{Wallach2002})}} Let $\mc{H} = \otimes_{i=1}^n \mc{H}_i$ with $\mathrm{dim}(\mc{H}_i) \geq 3$ finite for all $i \in \{1,\cdots,n\}$, $n \in \mathbb{N}$. For every non-negative, unentangled frame function $f: \sigma(\mc{H}) \rightarrow \mathbb{R}$
    there exists a self-adjoint operator $t: \mc{H} \rightarrow \mc{H}$ such that $f(v) = \mathrm{tr}(t p_v) = \langle v | t | v \rangle$ for all $v \in \sigma(\mc{H})$, $p_v = |v\rangle\langle v| \in \mc{P}(\mc{H})$.
\end{theorem}

As it turns out, Thm.~\ref{thm: Wallach} fails for frame functions $f: \sigma(\mc{H}) \rightarrow \mathbb{R}$ with constraints $\sum_{j_1,\cdots,j_n=1}^{d_1,\cdots,d_n} f(v_{j_1,1} \otimes \cdots \otimes v_{j_n,n}) = W$, $d_i := \mathrm{dim}(\mc{H}_i)$ further restricted to product bases $\beta(\mc{H}) := \{(v_{j_1,1} \otimes \cdots \otimes v_{j_n,n})_{j_1,\cdots,j_n=1}^{d_1,\cdots,d_n} \mid (v_{j_i,i})_{j_i=1}^{d_i} \in \mathrm{ONB}(\mc{H}_i)\}$ as shown in Prop.~5 in \cite{Wallach2002}. Nevertheless, a similar result \emph{does} hold for frame functions over twisted product bases $T(\beta(\mc{H})):= \{B' \in \mathrm{ONB}(H) \mid \exists B \in \beta(\mc{H}):\ B' \sim B \}$, where $\sim$ denotes the equivalence relation on unentangled bases under local unitary transformations.

\begin{proposition}\label{prop: Wallach for twisted product bases}
    Let $\mc{H} = \otimes_{i=1}^n \mc{H}_i$, $\mathrm{dim}(\mc{H}_i) \geq 3$ finite for all $i \in \{1,\cdots,n\}$, $n \in \mathbb{N}$. For every non-negative frame function over twisted product bases $f: \sigma(\mc{H}) \rightarrow \mathbb{R}$
    there exists a self-adjoint operator $t: \mc{H} \rightarrow \mc{H}$ such that $f(v) = \mathrm{tr}(t p_v) = \langle v |t| v \rangle$ for all $v \in \sigma(\mc{H})$, $p_v = |v\rangle\langle v| \in \mc{P}(\mc{H})$.
\end{proposition}

\begin{proof}
    In the proof of Thm.~2 in \cite{Wallach2002} replace unentangled bases with twisted product bases in the inductive hypothesis. The case $n=1$ still holds by Thm.~\ref{thm: Gleason theorem for frame functions}. Consider $\mc{H} = \mc{H}_1 \otimes V$, $V = \otimes_{i=2}^n \mc{H}_i$ with $\mathrm{dim}(\mc{H}_i) \geq 3$ for all $1 \leq i \leq n$. If $(v_j)_{j=1}^{d_1} \in \mathrm{ONB}(\mc{H}_1)$ is an orthonormal basis of $\mc{H}_1$ and $(u^j_{k})_{k=1}^{d_V} \in T(\beta(V))$ is a twisted product basis for $V$, then $(v_j \otimes u^j_k)_{j,k=1}^{d_1,d_V} \in T(\beta(\mc{H}))$ is a twisted product basis for $\mc{H}$. This follows since we can transform $u^j_k$ for every $j$ into a product basis on $V$ by the assumption that $u^j_k \in T(\beta(V))$, and the fact that applying local unitaries on subspaces $\sum_{j_i=1}^{d_i} v_{j_r,r} \otimes v_{j_i,i}$ for all $i$, we can transform between product bases in $\beta(V)$.
    
    Since $f$ is a twisted product frame function (of weight $W \in \mathbb{R}^+$), the function $f_v(u) = f(v \otimes u)$ is a non-negative twisted product frame function on $V$ (of weight $W_v = W_{v_1} = W - \sum_{j=2,k=1}^{d_1,d_V} f(v_j \otimes u^j_k) \in \mathbb{R}^+$) for each $v \in \mc{H}_1$. By the inductive hypothesis we thus find $f_v(u) = \langle u| t_V(v)|u\rangle$ for all $u \in \sigma(V)$ with $t_V(v): V \rightarrow V$ self-adjoint.
    
    Conversely, let $(u_k)_{k=1}^{d_V} \in T(\beta(V))$ be a twisted product basis for $V$ and $(v^k_j)_{j=1}^{d_1} \in \mathrm{ONB}(\mc{H}_1)$ for every $k$, then $(v^k_j \otimes u_k)_{j,k=1}^{d_1,d_V} \in T(\beta(\mc{H}))$ is a twisted product basis for $\mc{H}$ (by a similar argument as before), and by the inductive hypothesis we conclude $f_u(v) := f(v \otimes u) = \langle v|t_{\mc{H}_1}(u)|v \rangle$ for all $v \in S(\mc{H}_1)$ with $t_{\mc{H}_1}(u): \mc{H}_1 \rightarrow \mc{H}_1$ self-adjoint. The remainder of the proof proceeds as for Thm.~\ref{thm: Wallach} in \cite{Wallach2002}.
\end{proof}

One might expect (non-negative) unentangled frame functions to correspond with (non-negative) twisted frame functions. However, as a consequence of the failure of Keller's cube-tiling conjecture, which was proven in \cite{LagariasShor1996}, this turns out not to be the case.

\begin{proposition}\label{prop: strict inclusion}
     $T(\beta(\mc{H})) \subsetneq \mathrm{ONB}(\sigma(\mc{H}))$
\end{proposition}

\begin{proof}
    Clearly, every twisted product basis is also an unentangled basis. The fact that the other direction fails is non-trivial, but can be concluded from a counterexample to Keller's tiling conjecture \cite{LagariasShor1996}: for $n \geq 10$ construct the following tiling of $\mathbb{R}^n$ by cubes of length $2$ such that
    \begin{itemize}
        \item [(a)] the centers of all cubes are in $\mathbb{Z}^n$,
        \item [(b)] the tiling is $4\mathbb{Z}^n$-periodic,
        \item [(c)] no two cubes have a complete facet in common.
    \end{itemize}
    More precisely, let $C := \{(x_1,\cdots,x_n) \mid -1 \leq x_i \leq 1 \quad \forall i \in \{1,\cdots,n\}\}$ denote a cube (of length $2$). Then a tiling corresponds to $2^n$ equivalence classes of translates of $C$ of the form $\mathbf{m} + C + 4\mathbb{Z}^n$ for
    \begin{equation}\label{eq: cube tiling locations}
        \mathbf{m} = (m_1,\cdots,m_n) \in \mathbb{Z}^n, \quad 0 \leq m_i \leq 3\; .
    \end{equation}
    Next, consider the conditions: (i) $\mathbf{m}$ and $\mathbf{m'}$ have some $|m_i - m_i'| = 2$ and (ii) $\mathbf{m}$ and $\mathbf{m'}$ differ in two coordinate directions. Finally, denote by $G_n$ and $G_n^*$ two graphs, each of which has $4^n$ vertices labeled by the $4^n$ vectors in Eq.~(\ref{eq: cube tiling locations}) and $G_n$ has an edge between vertices $\mathbf{m}$ and $\mathbf{m'}$ if (i) holds, while $G_n^*$ is defined to have an edge between vertices $\mathbf{m}$ and $\mathbf{m'}$ if (i) and (ii) hold.
    
    Then a set $\mc{S}$ of $2^n$ vectors of the form in Eq.~(\ref{eq: cube tiling locations}) yields a $4\mathbb{Z}^n$-periodic cube tiling if and only if $\mc{S}$ forms a clique in $G_n$; it yields a $4\mathbb{Z}^n$-periodic cube tiling with no two cubes having a complete facet in common if and only if $\mc{S}$ forms a clique in $G_n^*$.
    
    We now translate this into a basis of $\mc{H} = (\mathbb{C}^2)^{\otimes 10}$. Consider the qubit states $|0\rangle$, $|1\rangle$, $|+\rangle := \frac{1}{\sqrt{2}}(|0\rangle + |1\rangle)$ and $|-\rangle := \frac{1}{\sqrt{2}}(|0\rangle - |1\rangle)$ and define the correspondence $m_i \mapsto |\psi(m_i)\rangle$ as follows:
    \begin{align*}
        |\psi(m_i = 0)\rangle &= |0\rangle_i & |\psi(m_i = 1)\rangle &= |+\rangle_i \\
        |\psi(m_i = 2)\rangle &= |1\rangle_i & |\psi(m_i = 3)\rangle &= |-\rangle_i
    \end{align*}
    First, note that $|\psi(\mc{S})\rangle := \{|\psi(m_1)\rangle \otimes \cdots \otimes |\psi(m_n)\rangle \mid \mathbf{m} \in \mc{S}\}$ forms a basis of $(\mathbb{C}^2)^{\otimes 10}$: there are $2^{10}$ vectors and it is easily seen that $\langle \psi(\mathbf{m})|\psi(\mathbf{m'}) \rangle = 0$ for $\mathbf{m}, \mathbf{m'} \in \mc{S}$, $\mathbf{m} \neq \mathbf{m'}$ by condition (i) above. Moreover, $|\psi(\mc{S})\rangle \in \sigma((\mathbb{C}^2)^{\otimes 10})$ is an unentangled basis by construction. However, any two vectors $|\psi(\mathbf{m})\rangle$, $|\psi(\mathbf{m'})\rangle$ for $\mathbf{m}, \mathbf{m'} \in \mc{S}$, $\mathbf{m} \neq \mathbf{m'}$ differ on at least two sites by condition (ii). It follows that no two-dimensional subspace of the form $x_{j_r,r} \otimes (v_{j_i,i} + v_{j'_i,i})$ is spanned by vectors in $|\psi(\mc{S})\rangle$. Yet, any twisted product basis has at least one two-dimensional subspace of this form, hence, $|\psi(\mc{S})\rangle$ cannot be a twisted product basis.
\end{proof}

Finally, we relate non-negative frame functions over twisted product bases with non-negative, non-signalling frame functions over product bases. Recall that a frame function over product bases is called non-signalling if for $i \in \{1,\cdots,n\}$ with $(v_{j_i,i})_{j_i=1}^{d_i}, (w_{k_i,i})_{k_i=1}^{d_i} \in \mathrm{ONB}(\mc{H}_i)$ and $x_{l_r,r} \in \mc{S}(\mc{H}_r)$ for all $l_r \in \{1,\cdots,d_r\}$, $r \neq i$ the following condition is satisfied (cf. Eq.~(\ref{eq: no-signalling constraints})):
\begin{equation*}
    \sum_{j_i=1}^{d_i} f(x_{l_r,r} \otimes v_{j_i,i}) = \sum_{k_i=1}^{d_i} f(x_{l_r,r} \otimes w_{k_i,i})
\end{equation*}

\begin{lemma}\label{lm: product bases + ns}
    Let $\mc{H} = \otimes_{i=1}^n \mc{H}_i$, $\mathrm{dim}(\mc{H}_i) \geq 3$ finite for all $i \in \{1,\cdots,n\}$, $n \in \mathbb{N}$. There is a bijective correspondence between non-negative, non-signalling frame functions $f: \sigma(\mc{H}) \rightarrow \mathbb{R}$ over product bases and non-negative frame functions over twisted product bases.
\end{lemma}

\begin{proof}
    Let $x_{l_r,r} \in \mc{S}(\mc{H}_r)$ for all $l_r \in \{1,\cdots,d_r\}$, $r \neq i$ and $(v_{j_i,i})_{j_i=1}^{d_i},(w_{k_i,i})_{j_i=1}^{d_i} \in \mathrm{ONB}(\mc{H}_i)$ such that w.l.o.g. $|v_{1,i}\rangle \langle v_{1,i}| + |v_{2,i}\rangle \langle v_{2,i}| = |w_{1,i}\rangle \langle w_{1,i}| + |w_{2,i}\rangle \langle w_{2,i}|$ and $v_{j_i,i} = w_{k_i,i}$ for $3 \leq j_i=k_i \leq n$. By no-signalling in Eq.~(\ref{eq: no-signalling constraints}),
    \begin{align*}
        f(x_{l_r,r} &\otimes v_{1,i}) + f(x_{l_r,r} \otimes v_{2,i}) \\
        =\ &\sum_{j_i=1}^{d_i} f(x_{l_r,r} \otimes v_{j_i,i}) - \sum_{j_i=3}^{d_i} f(x_{l_r,r} \otimes v_{j_i,i}) \\
        =\ &\sum_{k_i=1}^{d_i} f(x_{l_r,r} \otimes w_{k_i,i}) - \sum_{k_i=3}^{d_i} f(x_{l_r,r} \otimes w_{k_i,i}) \\
        =\ &f(x_{l_r,r} \otimes w_{1,i}) + f(x_{l_r,r} \otimes w_{2,i})\; .
    \end{align*}
    As twisted product bases are generated from local unitaries acting on two-dimensional subspaces of the form $x_{l_r,r} \otimes (v_{j_i,i} + v_{j'_i,i})$, $f$ is also a frame function over twisted product bases. Conversely, for the latter Eq.~(\ref{eq: no-signalling constraints}) holds since it holds already for two-dimensional subspaces.
\end{proof}

Clearly, a non-negative, non-signalling frame function of weight $1$ corresponds with a locally quantum, non-signalling probability distribution as defined in \cite{KlayRandallFoulis1987,Wehner2010}; Prop.~\ref{prop: Wallach for twisted product bases}, Prop.~\ref{prop: strict inclusion}, and Lm.~\ref{lm: product bases + ns} thus establish the precise relation between these different concepts, in particular, Prop.~\ref{prop: Wallach for twisted product bases} and Lm.~\ref{lm: product bases + ns} prove Thm.~\ref{thm: non-signalling frame functions correspond to self-adjoint operators PVM}.

\section{Systems of dimension two}\label{sec: Systems of dimensions two}

Note that Thm.~\ref{thm: non-signalling frame functions correspond to self-adjoint operators PVM} only applies to finite local dimensions $\mathrm{dim}(\mc{H}_i) \geq 3$. This restriction is due to Thm.~\ref{thm: Gleason theorem for frame functions}, since frame functions in two dimensions do not restrict to the latter.
Nevertheless, generalisations of Thm.~\ref{thm: Gleason theorem for frame functions} to two dimensions exist based on (subsets of) positive operator-valued measures (POVMs) \cite{Busch2003, CavesFuchs2004, WrightWeigert2018}.

More precisely, non-negative frame functions $f: \mc{E}(\mc{H}) \rightarrow \mathbb{R}^+$ of weight $W \in \mathbb{R}^+$ with domain $\mc{E}(\mc{H})$ the set of all \emph{effects}, i.e., convex combinations of projections, and such that $\sum_{i \in I} f(e_i) = W$ whenever
$\sum_{i \in I} e_i = 1$, correspond to density matrices: $f(e) = W \mathrm{tr}(\rho e)$ for all $e \in \mc{E}(\mc{H})$ and $\mathrm{dim}(\mc{H}) \geq 2$ finite.\\
\indent Similarly, replacing $\sigma(\mc{H})$ by $\sigma(\mc{E}(\mc{H}))$ (equivalently, projection-valued measures (PVMs) by POVMs) in the otherwise analogous definitions of (twisted) product frame functions and no-signalling in Eq.~(\ref{eq: no-signalling constraints}), one obtains a generalisation for systems with $\mathrm{dim}(\mc{H}_i) = 2$.

\begin{theorem}\label{thm: non-signalling frame functions correspond to self-adjoint operators (effects)}
    Let $\mc{H} = \bigotimes_{i=1}^n \mc{H}_i$, $\mathrm{dim}(\mc{H}_i) \geq 2$ finite for all $i \in \{1,\cdots,n\}$, $n \in \mathbb{N}$. For every non-negative, non-signalling frame function over product POVMs $f: \sigma(\mc{E}(\mc{H})) \rightarrow \mathbb{R}$
    there exists a self-adjoint operator $t: \mc{H} \rightarrow \mc{H}$ such that $f(e) = \mathrm{tr}(t e)$ for all $e \in \sigma(\mc{E}(\mc{H}))$.
\end{theorem}

\begin{proof}
    By \cite{Busch2003} frame functions over $\sigma(\mc{E}(\mc{H}))$ correspond to quantum states for every $\mc{H}_i$ in $\mc{H} = \bigotimes_{i=1}^n \mc{H}_i$ with $\mathrm{dim}(\mc{H}_i) \geq 2$ finite. With this the inductive proof of Prop.~\ref{prop: Wallach for twisted product bases} goes through also for systems of local dimension two.
    The same holds for the correspondence of the no-signalling condition in Eq.~(\ref{eq: no-signalling constraints}) and constraints on $f$ arising from local transformations leaving convex combinations $\sum_i e_i \leq 1$ invariant in Lm.~\ref{lm: product bases + ns}.
\end{proof}

\section{Gleason's theorem in contextual form}\label{sec: Gleason's theorem in contextual form}

In \cite{Doering2005, Doering2012b} it was shown that Gleason's theorem can be reformulated in terms of global sections over the probabilistic presheaf defined as follows.

\begin{definition}\label{def: simple probabilistic presheaf}
    Let $\mc{V}(\mc{H})$ be the partial order of contexts over the Hilbert space $\mc{H}$. The functor $\PPi: \mc{V}(\mc{H})^\mathrm{op} \rightarrow \mathbf{Set}$,
    \begin{align*}
        \PPi(V) &:= \{\mu_V: \mc{P}(V) \rightarrow [0,1] \mid \mu_V(1) = 1,\\
        \forall p,p' \in &\mc{P}(V),pp'=0:\ \mu_V(p+p') = \mu_V(p)+\mu_V(p') \}\; , \\
        \PPi(\tilde{V} \subseteq V)&: \PPi_V \rightarrow \PPi_{\tilde{V}},\ \mu_V \mapsto \mu_{\tilde{V}} := \mu_V|_{\tilde{V}}\; ,
    \end{align*}
    is called the \emph{probabilistic presheaf $\PPi(\mc{V}(\mc{H}))$ over $\mc{H}$}.
\end{definition}

Note that Def.~\ref{def: simple probabilistic presheaf} differs from Def.~\ref{def: dilated probabilistic presheaf} in that we do not require probability distributions in contexts to admit dilations. For single systems, the two definitions are closely related. To see this, note that every density matrix $\rho: \mc{H} \rightarrow \mc{H}$ can be purified, i.e., there exists a Hilbert space $\mc{K} = \mc{H}' \otimes \mc{H}$ and a vector $|\psi_\rho\rangle \in \mc{K}$ such that $\rho =  \mathrm{tr}_{\mc{H}'}(|\psi_\rho\rangle\langle\psi_\rho|)$.
In particular, for pure states we find that $\mc{K} = \mc{H}$, $\rho = |\psi_\rho\rangle\langle\psi_\rho|$, and $\varphi^\gamma: \mc{P}(\mc{H}) \rightarrow \mc{P}(\mc{H})$ is the identity map such that $\gamma(p) = \mathrm{tr}_{\mc{H}}(\rho p) = \mathrm{tr}_{\mc{K}}(|\psi_\rho\rangle\langle\psi_\rho| \phi^\gamma(p)) = \langle\psi_\rho| p |\psi_\rho\rangle$. More generally, for mixed states we may therefore interpret the additional constraint on the existence of dilations (purifications) in contexts as a consistency condition on convex combinations of pure states.

For single systems, the set of global sections of the dilated probabilistic presheaf in Def.~\ref{def: dilated probabilistic presheaf} therefore coincides with the set of global sections of the probabilistic presheaf in Def.~\ref{def: simple probabilistic presheaf}; in both cases we recover the (quantum) state space by \cite{Doering2012b}. On the other hand, Def.~\ref{def: dilated probabilistic presheaf} differs from Def.~\ref{def: simple probabilistic presheaf} over composite systems, where the additional convexity condition (implicit in the dilations in contexts) rules out certain global sections, which correspond to non-positive linear operators between the algebras of the component systems (via the Choi-Jamio\l kowski isomorphism).

\section{Proof of Proposition \ref{prop: Bell-type Gleason for global sections}}\label{sec: Proof of Theorem 2}

Prop.~\ref{prop: Bell-type Gleason for global sections} follows with Prop.~\ref{prop: Wallach for twisted product bases} and the following theorem.

\begin{theorem}\label{thm: frame functions and probabilistic presheaf}
    Global sections of the probabilistic presheaf in Def.~\ref{def: simple probabilistic presheaf} over product contexts $\gamma \in \Gamma(\PPi(\mc{V}(\mc{H}_1) \times \mc{V}(\mc{H}_2)))$ with local dimension $\mathrm{dim}(\mc{H}_i) \geq 3$ finite, bijectively correspond to non-negative frame functions of weight $1$ over twisted product bases $f^\gamma: \sigma(\mc{H}) \rightarrow \mathbb{R}$.
\end{theorem}

\begin{proof}
	Every frame function over twisted product bases defines a global section on product contexts by $\gamma_f(p_{v_1} \otimes p_{v_2}) = f(v_1 \otimes v_2)$, $p_{v_i} = |v_i\rangle\langle v_i|$ for all $v_i \in \sigma(\mc{H}_i)$. Marginalisation over product contexts corresponds to no-signalling and thus follows from the constraints on $f$ over twisted product bases by Lm.~\ref{lm: product bases + ns}.
	
	Conversely, every global section $\gamma \in \Gamma(\PPi(\mc{V}(\mc{H}_1) \times \mc{V}(\mc{H}_2)))$ over product contexts $V \in \mc{V}(\mc{H}_1) \times \mc{V}(\mc{H}_2)$ defines a map $f^\gamma: \sigma(\mc{H}) \rightarrow \mathbb{R}^+$, $f^\gamma(v \otimes w) := \gamma(p_v \otimes q_w)$ via the link between product projections $p_v \otimes q_w \in \mc{P}(\mc{H})$ and basis elements $v \otimes w \in \sigma(\mc{H})$ given by $p_v = |v\rangle\langle v|$, $q_w = |w\rangle\langle w|$.
	Moreover, this map satisfies the constraints encoded in twisted product bases, which for global sections arise from marginalisation between product contexts of the form (and by symmetry for $i=1 \leftrightarrow i=2$):
	
	\begin{figure}[!htb]
    \centering
    \vspace{-0.3cm}
    \begin{minipage}{.15\textwidth}
        \centering
        \begin{align*}
	        V &:= V_1 \times \{p_{1,2},p_{2,2},(p_{1,2}+p_{2,2})^\perp\} \\
	        \dot{V} &:= V_1 \times \{q_{1,2},q_{2,2},(p_{1,2}+p_{2,2})^\perp\} \\
	        \tilde{V} &:= V_1 \times \{(p_{1,2}+p_{2,2}),(p_{1,2}+p_{2,2})^\perp\}
	\end{align*}
    \end{minipage}%
    \begin{minipage}{0.15\textwidth}
        \centering
        \vspace{0.3cm}
        \begin{tikzpicture}[node distance=0.6cm, every    node/.style={scale=0.9}]
            \node(V)                    {$V$};
            \node(0)    [right of=V]  {};
            \node(W)    [right of=0]  {$\dot{V}$};
            \node(X)    [below=0.65cm of 0]  {$\tilde{V}$};

            \draw [right hook-latex](X) -- (V);
            \draw [left hook-latex](X) -- (W);
        \end{tikzpicture}
    \end{minipage}
    \end{figure}
    
	Here, we define contexts via their projections $p_{j_i,i} := |v_{j_i,i}\rangle \langle v_{j_i,i}|$, $q_{k_i,i} := |w_{k_i,i}\rangle \langle w_{k_i,i}|$ corresponding to product bases $(v_{j_i,i})_{j_i=1}^{d_i},(w_{k_i,i})_{k_i=1}^{d_i} \in \mathrm{ONB}(\mc{H}_i)$ such that $V_1 = \{p_{1,1},\cdots,p_{d_1,1}\}$ and $p_{1,2} + p_{2,2} = q_{1,2} + q_{2,2}$. Again, this is analogous to the proof of Lm.~\ref{lm: product bases + ns}.
\end{proof}

Thm.~\ref{thm: non-signalling frame functions correspond to self-adjoint operators PVM} for non-signalling frame functions therefore directly translates to Thm.~\ref{thm: frame functions and probabilistic presheaf} for global sections of the probabilistic presheaf. Clearly, this implies Prop.~\ref{prop: Bell-type Gleason for global sections}. However, note that it is no longer the case that every non-negative frame function over twisted product bases of weight $1$ also corresponds to a global section of the Bell presheaf. The reason is that the constraints on probability distributions between product contexts are more restrictive in the case of the dilated probabilistic presheaf than those in the probabilistic presheaf, which in turn allows us the prove the stronger result in Lm.~\ref{lm: global sections to Jordan homos}.

\section{Proof of Lemma \ref{lm: global sections to Jordan homos}}\label{sec: Proof of Lm: global sections to Jordan *-homos}

Let $\mc{V}_{1\& 2} := \mc{V}(\mc{H}_1) \times \mc{V}(\mc{H}_2)$ and $\gamma \in \Gamma(\PPi(\mc{V}_{1\& 2}))$.\\

\textbf{Step 1 - Gleason's theorem.}
Fix a context $V_1 \in \mc{V}(\mc{H}_1)$ and consider the corresponding partial order of contexts inherited from $\mc{V}_{1\&2}$ by restriction,
\begin{equation*}
    \mc{V}_{1\&2}(V_1) := \{V_1 \times V_2 \mid V_2 \in \mc{V}(\mc{H}_2)\}\; .
\end{equation*}
In every context $V = V_1 \times V_2 \in \mc{V}_{1\&2}$, the probability distribution $\mu^\gamma_V \in \PPi(\mc{V}_{1\&2})_V$ corresponding to the global section $\gamma$ takes the form: $\forall p \in \mc{P}(V_1), q \in \mc{P}(V_2)$,
\begin{equation}\label{eq: Gleason 1}
    \mu^\gamma_V(p,q) = \mu_{V_1}^\gamma(p) \mu_{V_2}^{\gamma}(q \mid p) = \mu_{V_1}^\gamma(p) \gamma_2^p(q)\; .
\end{equation}
Here, $(\mu_{V_2}^\gamma (\underline{\ \ } \mid p))_{V_2 \in \mc{V}(\mc{H}_2)} =: \gamma_2^p \in \Gamma(\PPi(\mc{V}_{1\&2}(V_1)))$ is a global section of the probabilistic presheaf $\PPi(\mc{V}_{1\&2}(V_1))$, which also depends on $p \in \mc{P}(V_1)$. Since $\PPi(\mc{V}_{1\&2}(V_1)) \cong \PPi(\mc{V}(\mc{H}_2))$, by Gleason's theorem there is a unique density matrix $\rho_2^p$ such that $\gamma_2^p = \mathrm{tr}(\rho_2^p \cdot \underline{\ \ }) \in \mc{S}(\mc{L}(\mc{H}_2))$.
Moreover, since $V_1 \in \mc{V}(\mc{H}_1)$ is arbitrary, Eq.~(\ref{eq: Gleason 1}) holds for all $p \in \mc{P}(\mc{H}_1)$. Let $p = p_1 + p_2$ with $p_1,p_2 \in \mc{P}(\mc{H}_1)$ orthogonal, i.e., $p_1p_2 = 0$. As $\gamma$ is additive we have,
\begin{equation*}
    \mu_{V_1}^\gamma(p) \gamma_2^p = \mu_{V_1}^\gamma(p_1) \gamma_2^{p_1} + \mu_{V_1}^\gamma(p_2) \gamma_2^{p_2}\; .
\end{equation*}
It follows that the map $\varrho^\gamma(p) := \mu_{V_1}^\gamma(p) \rho_2^p$ also satisfies,
\begin{equation*}\label{eq: finite additive operator measure}
    \varrho^\gamma(p) = \varrho^\gamma(p_1) + \varrho^\gamma(p_2)
\end{equation*}
for all $p = p_1 + p_2$ with $p_1,p_2 \in \mc{P}(\mc{H}_1)$, $p_1p_2 = 0$. Since $\mu^\gamma_{V_1}(p) \geq 0$ for all $p \in \mc{P}(\mathbb{C}^{d_1})$, $\varrho^\gamma: \mc{P}(\mc{H}_1) \rightarrow \mc{L}(\mc{H}_2)_+$ defines an additive map into the positive matrices on $\mc{H}_2$.

\vspace{0.5cm}

\textbf{Step 2 - Naimark's theorem.}
Note that $\varrho^\gamma(p)$ corresponds to a density matrix (up to the factor $\mu^\gamma_{V_1}(p)$). We can thus find a Hilbert space $\mc{K} = \mc{H}' \otimes \mc{H}_2$ and a purification $|\psi\rangle^{\gamma,p} \in \mc{K}$ such that
$\varrho^\gamma(p) = \mathrm{tr}_{\mc{H}'}(|\psi\rangle^{\gamma,p}\langle\psi|^{\gamma,p})$.
This is a special case of Naimark's theorem \cite{Naimark1943}. In fact, by Gelfand duality every context $V_1 \in \mc{V}(\mc{H}_1)$ defines a compact Hausdorff space whose $\sigma$-algebra of open (and closed) sets corresponds to the projection lattice $\mc{P}(V_1)$. $\varrho^\gamma|_{V_1}$ thus becomes a positive operator-valued measure, for which
Naimark's theorem guarantees the existence of a Hilbert space $\mc{K}$, a linear map $v: \mc{H}_2 \rightarrow \mc{K}$, and an embedding (a spectral measure) $\varphi_{V_1}: \mc{P}(V_1) \hookrightarrow \mc{P}(\mc{K})$ such that $\varrho^\gamma|_{V_1} = v^* \varphi^\gamma_{V_1} v$. Now, note that by the definition of the dilated probabilistic presheaf (cf. Def.~\ref{def: dilated probabilistic presheaf}), we obtain such embeddings \emph{consistently across contexts} $V_1 \in \mc{V}(\mc{H}_1)$, i.e., the map $\varrho^\gamma_{V_1}
= v^* \varphi^\gamma_{V_1} v$ defines a collection $\varphi^\gamma = (\varphi^\gamma_{V_1})_{V_1 \in \mc{V}(\mc{H}_1)}$ with $\varphi^\gamma_{\tilde{V_1}} = \varphi^\gamma_{V_1}|_{\tilde{V_1}}$ for all $\tilde{V_1} \subset V_1 \in \mc{V}(\mc{H}_1)$. In particular, we can choose $\mc{K}$ and $v: \mc{H}_2 \rightarrow \mc{K}$ independently of contexts.

Hence, $\varphi^\gamma: \mc{P}(\mc{H}_1) \rightarrow \mc{P}(\mc{K})$ is a kind of \emph{globally} defined purification of $\varrho^\gamma$. More precisely, $\varphi^\gamma$ is an orthomorphism, i.e., for all $p,p' \in \mc{P}(\mc{H}_1)$, $pp'=0$ we have $\varphi^\gamma(0) = 0$, $\varphi^\gamma(1-p) = 1- \varphi^\gamma(p)$, and $\varphi^\gamma(p + p') = \varphi^\gamma(p) + \varphi^\gamma(p')$.
This follows immediately since $\varphi^\gamma = (\varphi^\gamma_{V_1})_{V_1 \in \mc{V}(\mc{H}_1)}$ is a family of embeddings (spectral measures). In particular, there is an embedding $\varphi^\gamma_{V_1}$ with $p,p' \in \mc{P}(V_1)$.

\vspace{0.5cm}

\textbf{Step 3 - Dye's theorem (Gleason's theorem II).}
By a variant of Dye's theorem in \cite{BunceWright1993}, $\varphi^\gamma$ further lifts to a Jordan homomorphism $\Phi^\gamma: \mc{J}(\mc{H}_1) \rightarrow \mc{J}(\mc{K})$. (Recall that the linear map $\Phi^\gamma$ is a Jordan homomorphism if it preserves the anticommutator, $\{a,b\} := ab + ba$ for all $a,b \in \mc{J}(\mc{H}_1)$.)

There are two steps to this theorem. First, we obtain a quasilinear map by extending $\varphi^\gamma$ linearly in contexts. By the generalisation of Gleason's theorem in \cite{BunceWright1992}, this map lifts to a linear map $\Phi^\gamma: \mc{J}(\mc{H}_1) \rightarrow \mc{J}(\mc{K})$. In particular, $\varrho^\gamma$ uniquely extends to the linear map $\phi^\gamma = v^* \Phi^\gamma v$. Second, since $\varphi^\gamma$ is an orthomorphism, $\Phi^\gamma$ preserves squares. Namely, for every $a \in \mc{L}_\mathrm{sa}(\mc{H}_1)$ with spectral decomposition $a = \sum_{i=1}^n a_ip_i$ one computes,
\begin{equation*}
    \Phi^\gamma(a^2) = \Phi^\gamma(\sum_{i=1}^n a_i^2p_i) = \sum_{i=1}^n a_i^2 \Phi^\gamma(p_i) = \Phi^\gamma(a)^2\; .
\end{equation*}
Moreover, with $\{a,b\} = ab + ba = \frac{1}{2}[(a+b)^2 - a^2 - b^2]$,
\begin{align*}
    \Phi^\gamma(\{a,b\}) &= \Phi^\gamma(\frac{1}{2}[(a+b)^2 - a^2 - b^2]) \\
    &= \frac{1}{2}[(\Phi^\gamma(a)+\Phi^\gamma(b))^2 - \Phi^\gamma(a)^2 - \Phi^\gamma(b)^2] \\
    &= \{\Phi^\gamma(a),\Phi^\gamma(b)\}\; .
\end{align*}
Hence, $\Phi^\gamma$ is a Jordan homomorphism. Finally, we extend $\Phi^\gamma$ to a Jordan $*$-homomorphism on the complexified algebras $\mc{J}(\mc{H}_1) = \mc{J}_\mathrm{sa}(\mc{H}_1) + i\mc{J}_\mathrm{sa}(\mc{H}_1)$, $\mc{J}(\mc{K}) = \mc{J}_\mathrm{sa}(\mc{K}) + i\mc{J}_\mathrm{sa}(\mc{K})$ by setting $\Phi^\gamma(a + ib) := \Phi^\gamma(a) + i\Phi^\gamma(b)$ for all $a,b \in \mc{L}_\mathrm{sa}(\mc{H}_1)$. Then, $\Phi^\gamma((a + ib)^*) = \Phi^\gamma(a+ib)^*$. In summary, every global section of the Bell presheaf $\gamma \in \Gamma(\PPi(\mc{V}_{1\& 2}))$ uniquely lifts to a positive linear map $\phi^\gamma: \mc{L}(\mc{H}_1) \rightarrow \mc{L}(\mc{H}_2)$, with $\phi^\gamma = v^* \Phi^\gamma v$ for some linear map $v: \mc{H}_2 \rightarrow \mc{K}$ and $\Phi^\gamma: \mc{J}(\mc{H}_1) \rightarrow \mc{J}(\mc{K})$ a Jordan $*$-homomorphism. This completes the proof.

\section{Proof of Theorem \ref{thm: Bell-type Gleason}}\label{sec: Proof of Theorem 3}

\begin{proof}
    Every state $\sigma \in \mc{S}(\mc{L}(\mc{H})) = \mc{S}(\mc{L}(\mc{H}_1) \otimes \mc{L}(\mc{H}_2))$ defines a completely positive map $\phi^\sigma: \mc{L}(\mc{H}_1) \rightarrow \mc{L}(\mc{H}_2)$ via the Choi-Jamio\l kowski isomorphism. It thus follows that $\sigma$ defines a unique orientation-preserving global section of the composite presheaf
    $\gamma^\sigma(p,q) = \mu^\gamma_V(p,q) := \sigma(p \otimes q)$
    (by restriction to contexts $V \in \mc{V}(\mc{H}_1) \times \mc{V}(\mc{H}_2)$).
    
    For the converse direction, let $\gamma \in \Gamma(\PPi(\widetilde{\mc{V}(\mc{H}_1)} \times \widetilde{\mc{V}(\mc{H}_2)}))$. By Lm.~\ref{lm: global sections to Jordan homos}, $\phi^\gamma = v^* \Phi^\gamma v$ for $\Phi^\gamma$ a Jordan $*$-homomorphism, and since $\gamma$ is also orientation-preserving, 
    $\Phi^\gamma$ further lifts to a $*$-homomorphism. By Stinespring's theorem $\phi^\gamma$ is thus completely positive and by Choi's theorem $\rho_{\phi^\gamma}$ is a density matrix.
    Hence, $\sigma^\gamma := \mathrm{tr}(\rho_{\phi^\gamma} \cdot \underline{\ \ }) \in \mc{S}(\mc{L}(\mc{H})) = \mc{S}(\mc{L}(\mc{H}_1) \otimes \mc{L}(\mc{H}_2))$ defines a unique quantum state.
\end{proof}

\end{document}